\documentclass[onecolumn]{per}

\usepackage{amsmath,amscd,amssymb,amsgen,amsfonts,amsbsy}
\usepackage{mathrsfs}
\usepackage{bm}
\usepackage{bbm}
\usepackage{color}
\usepackage{textcomp}
\usepackage{float}
\usepackage{latexsym,graphicx}

\usepackage{url}

\newtheorem{lemma}{\textbf{Lemma}}

\newtheorem{theorem}{\textbf{Theorem}}
\newtheorem{remark}{\textbf{Remark}}
\newtheorem{corollary}{\textbf{Corollary}}

\begin{document}
\title{On the equivalence between multiclass processor sharing and random order scheduling policies}

\author{
  \begin{tabular}{ccc}
    Konstantin Avrachenkov  & Tejas Bodas \\
    INRIA, Sophia Antipolis & LAAS-CNRS, Toulouse
  \end{tabular}
}

\maketitle{}
\begin{abstract}

 {Consider a single server system serving a multiclass population.
Some popular scheduling policies for such system are the discriminatory processor sharing (DPS),
discriminatory random order service (DROS), generalized processor
sharing (GPS) and weighted fair queueing (WFQ). In this paper, we
propose two classes of policies, namely MPS (multiclass processor sharing)
and MROS (multiclass random order service), that generalize the four policies
mentioned above. For the special case when the multiclass population arrive according
to Poisson processes and have independent
and exponential service requirement with parameter $\mu$, we show that the
tail of the sojourn time distribution for a class $i$ customer in a
system with the MPS policy is a constant multiple of
the tail of the waiting time distribution of a class $i$  customer in a
system with the MROS policy. This result implies that for a class $i$ customer,
the tail of the sojourn time distribution in a system with the DPS (GPS)
scheduling policy is a constant multiple of the tail of the waiting time
distribution in a system with the DROS (respectively WFQ) policy.}

\end{abstract}

\section{Introduction}

 Consider a single server system with multiclass customers.
 Some commonly used scheduling policies in such
 multiclass system are DPS, DROS, GPS and WFQ.
 A quick overview of these policies is as follows.
Policies like GPS and DPS are variants of the processor sharing
policy  {where the server} can serve multiple customers from the system simultaneously.
In case of GPS, a separate queue is maintained for each customer
class and the total service capacity of the server is shared among customers
of the different classes in proportion to predefined weights $p_i.$
The GPS scheduling policy is often considered as a generalization of
the head-of-line processor sharing policy (HOLPS) as described in
\cite{Aalto07,Demers90,KMK}. (Refer \cite{Fayolle79,Morrison93} for details about
HOLPS). As a generalization of HOLPS, GPS maintains a FIFO scheduling
policy  {within the queue of each class} and only the head-of-line customers
of different classes are allowed to share the processor.
The share of the server for a  head-of-line Class~$i$ customer is
proportional to the weight $p_i$ and is independent of the number of
other customers in the queue. The service rate received by the customer
is precisely given by $\frac{p_i}{\sum_{j=1}^{N} p_j \phi_j}$ where $\phi_j = 1$ if
the queue has at least one class $j$ customer and  $\phi_j = 0$
otherwise. Refer  Parekh and Gallager \cite{Parekh93}, Zhang et al.
\cite{Zhang94} for an early analysis of the model.

In case of DPS, the total service capacity is shared among all the customers present in the
system and not just among the head-of-line customers of different classes.
The share of the server for a customer of a class is not only
in proportion to the class weight, but also depends on the number
of multiclass customers present in the queue. In particular,
a Class~$i$ customer in the system is served at a rate of
$\frac{p_i}{\sum_{j=1}^{N} p_j n_j}$ where $n_j$ denotes the
number of Class~$j$ customers in the system.
The DPS system  was first introduced by Kleinrock
\cite{Kleinrock67-2} and subsequently analyzed by several authors
\cite{Aetal05,Fayolle80,Haviv07-2,Haviv97,Kim04}.
See \cite{AAA06} for a survey of various results on DPS.

The DROS and WFQ scheduling policies are also characterized
by an associated weight for each customer class. However these
policies are not a variant of the processor sharing policies
and hence their respective server can only serve one customer
at a time. The DROS and WFQ policies differ in their exact rule for choosing
the next customer. In the DROS policy, the probability of
choosing a customer for service depends on the weights and the
number of customers of the different classes in the queue.
A Class~$i$ customer is thus chosen with a probability of
$\frac{p_i}{\sum_{j=1}^{N} p_j n_j}$ where $n_j$ denotes the
number of Class~$j$  {customers waiting in the system} for service.
DROS policy is also know as relative priority policy and was first
introduced by Haviv and wan der Wal \cite{Haviv97}. For more
analysis of this policy we refer to \cite{Haviv07,Kim11}.
In the WFQ policy, a separate queue for each class is maintained
and the next customer is chosen randomly from among the
head-of-line customers of different classes. As in case of the
GPS scheduling policy, a FIFO scheduling policy is used within
each queue for a class. WFQ can be seen as a packetised version
of GPS and the probability of choosing a head-of-line Class~$i$
customer for service is given by
$\frac{p_i}{\sum_{j=1}^{N} p_j \phi_j}$ where $\phi_j$ is as defined
earlier. Refer Demers \cite{Demers90} for the detailed analysis
of the WFQ policy.

It is interesting to note that for a Class~$i$ customer, the service rate
received {\it in DPS} and the probability of being chosen next for service
{\it in case of DROS} is given by $\frac{p_i}{\sum_{j=1}^{N} p_j n_j}.$ Similarly,
the service rate received {\it in GPS} and the probability of being chosen next
for service {\it in case of WFQ} is $\frac{p_i}{\sum_{j=1}^{N} p_j \phi_j}.$
This similarity in the scheduling rules motivates us to compare
the  {tail of the waiting time and sojourn time distributions}
of the multiclass customers
with these scheduling policies.  {Assuming identically distributed service
requirements for all customers}, we will show that the  {tail of the waiting time distribution} of a
Class~$i$ customer in a system with DROS (WFQ) scheduling
policy is $\rho$ times the  {tail of the sojourn time distribution} of any Class~$i$
customer with DPS (resp. GPS) scheduling policy. This is a generalization
of \cite{Borst03}, where the equivalence has been established between single-class
processor sharing and random order service discipline.

\textit{Organization:} In the next section,
we introduce a generalized notion of multiclass processor sharing (MPS)
and random order service (MROS) policies. The DPS, GPS, DROS
and WFQ policies will turn out to be special cases of MPS and MROS.
In Section \ref{sec:main}, we show that the  {tail of the sojourn time distribution}
of a Class $i$ customer with MPS scheduling is equivalent to the
 {tail of the waiting time distribution} of a Class~$i$
customer with MROS policy. As a special case, this proves the mentioned
equivalences among the four multiclass scheduling policies.

{\textit{Notation:}
 {We use $N$ to denote the total number of customer classes in the system.
Let $\lambda_i$ denote the arrival rate for a Class~$i$ customer,
$i= 1 \ldots N.$ Let $\sum_{i=1}^N \lambda_i = \Lambda,$ $\rho_i = \frac{\lambda_i}{\mu}$ and $\rho =
\sum_{i=1}^N \rho_i.$ Further, let $p_i$ denote a weight parameter associated with
a Class~$i$ customer.}

{\textit{Assumptions:}
 
{We assume that the service requirement of each customer is independent and
exponentially distributed with rate $\mu$.  Thus the service requirements are independent
of their class. We also assume that customers from different classes arrive according to
independent Poisson processes. For the purpose of stability, we assume
that $\Lambda < \mu.$}

\section{Generalized multiclass scheduling policies}
In this section, we will describe two multiclass scheduling policies that
are a generalization of policies such as DPS, DROS, GPS and WFQ. The two
policies are based on the processor sharing and random order service mechanism
and will be labeled as MPS and MROS respectively.

The MPS scheduling policy is a multiclass processor sharing policy where the server can
serve multiple customers simultaneously. A separate queue for each customer
class is maintained and a FIFO  {scheduling} policy is used within each queue
of a class.  {The MPS scheduling policy is
parameterized by a vector $\bar{\alpha}=(\alpha_1,\ldots,\alpha_N)$ that characterizes
the maximum permisible number of customers of each class that can be served
simultaneously with other customers. We shall henceforth use the notation MPS$(\bar{\alpha})$
when we talk about an MPS policy with parameter $\bar{\alpha}.$
Let $n_i$ denote the instantaneous number of Class~$i$ customers in the queue and
$\bar{n}:=(n_1,\ldots,n_N)$ denotes the corresponding state in
the MPS system.} Let $\beta_i(\bar{n})$ denote the number of
Class~$i$ customers under service when the state of the MPS system is $\bar{n}$. Then, clearly
$\beta_i(\bar{n}) = \mbox{~min~}(n_i,\alpha_i).$ In other words, if $n_i \leq \alpha_i,$ then all the
Class~$i$ customers present in the queue are being served simultaneously for
$i = 1, \ldots, N$. However if $n_i > \alpha_i,$
then only  the first $\alpha_i$ customers of Class~$i$ in its queue
are served simultaneously. It should be noted that due to the FIFO policy
within each queue of a class, only the first $\beta_i(\bar{n})$ customers in the
queue are served at any time.  To lighten some of the notation,
we shall drop the dependence on $\bar{n}$ and use only $\beta_i$
when the context is clear.   {For an MPS$(\bar{\alpha})$ scheduling policy in state
$\bar{n},$} the service rate received by a particular
Class~$i$ customer which is in service is given by
 $\frac{p_i}{\sum_{j=1}^{N} p_j \beta_j}.$  When $\alpha_i = \infty,$
 for $i=1 \mbox{~to~} N, $ the corresponding  scheduling policy will be
 denoted by MPS$(\bar{\infty}).$ In this case,
 $\beta_i = \mbox{~min~}(n_i,\infty) = n_i$ and therefore MPS$(\bar{\infty})$
 corresponds to the DPS scheduling policy. Similarly if
 $\bar{e}=(1,\ldots, 1),$ then MPS$(\bar{e})$ corresponds to the GPS
 scheduling policy where only the head-of-line customers of each class
 can be served.

In a similar manner, we can define the MROS$(\bar{\alpha})$ scheduling policy
where $\bar{\alpha}=(\alpha_1,\ldots,\alpha_N)$
denotes  {the vector of multiclass customers} from which the subsequent
customer is chosen for service. As in case of the MPS policy, note that
a separate FIFO queue for each customer class is also maintained for the MROS system.
At any given time, the first $\beta_i = \mbox{~min~}(n_i,\alpha_i)$ customers
are candidates for being chosen for service while
the remaining $n_i-\beta_i$ customers have to wait for their turn.
 {Note that the state $\bar{n}$ for an MROS$(\bar{\alpha})$ scheduling
policy denotes the vector of waiting multiclass customers present in the system.}
In the MROS$(\bar{\alpha})$ system, a Class~$i$ customer within the first $\beta_i$
customers in its queue will be chosen next for service with probability $\frac{p_i}{\sum_{j=1}^{N} p_j \beta_j}.$
As in case of the MPS scheduling, MROS$(\bar{\infty})$ corresponds to the DROS
policy whereas MROS$(\bar{e})$ corresponds to the WFQ policy.

\begin{remark}
A policy closely related to the MPS discipline is the
\textit{limited processor sharing} (LPS)
policy. LPS is a single class processor sharing policy
parametrized by an integer $c$ where $c$
denotes the maximum number of customers that can be
served simultaneously. Here $c = \infty$ corresponds to
the \textit{processor sharing} policy while $c = 1$ corresponds
to FCFS policy. LPS can also be viewed as a
special case of the MPS policy when there is a single
service class for the arriving customers. See \cite{Avi-Itzhak88,Zhang11} more more
details about the LPS-c policy.
\end{remark}

Having introduced the generalized multiclass scheduling policies, we
shall now establish an equivalence relation between the tail of the sojourn time distribution of a
Class~$i$ customer in MPS system with the tail of the waiting time distribution
of a Class~$i$ customer in MROS system.

\section{Comparing the sojourn and waiting time distributions in MPS and MROS}
\label{sec:main}
The analysis in this section is inspired from that in \cite{Borst03}
where a similar result is established for the case of a single class of
population.  {For a given state of
$\bar{n}= (n_1, \ldots, n_N),$ define $n:=\sum_{i=1}^N n_i.$}
Let random variable $\bm{S_i}(\bar{\alpha},\bar{n})$ denote the conditional
sojourn time experienced by an arriving Class~$i$ customer  {that sees the
$MPS(\bar{\alpha})$ system in state $\bar{n}$.} The corresponding unconditional random variable will be denoted by
$\bm{S_i}(\bar{\alpha}).$ We shall occasionally use the notation
$MPS(\bar{\alpha}, \bar{n})$ to denote the MPS$(\bar{\alpha})$ system with $\bar{n}$ customers.
Along similar lines, let the random variable $\bm{W_i}(\bar{\alpha},\bar{n})$ denote
waiting time (time until chosen for service) experienced by an arriving Class~$i$
customer that sees  {the MROS$(\bar{\alpha})$ system in state $\bar{n},$ i.e., it sees a vector of
$\bar{n}$ waiting customers in the system}. This system will be often denoted as
$MROS(\bar{\alpha}, \bar{n})$ and the unconditional random variable will
be denoted by $\bm{W_i}(\bar{\alpha}).$ Let $ \mathbb{P}$ and $\mathbb{P'}$ denote
the probability distribution of the random variables
$\bm{S_i}(\bar{\alpha},\bar{n})$ and $\bm{W_i}(\bar{\alpha},\bar{n})$ respectively.
(The dependence of these distributions on $\bar{n}$ have been suppressed for notational
convenience.) We now state the main result of this paper.

\begin{theorem}
$\rho P(\bm{{S}_i}(\bar{\alpha}) > t) =  P(\bm{{W}_i}(\bar{\alpha}) > t)$ for
$i = 1,\ldots,N.$
\end{theorem}
\begin{proof}
As in \cite{Borst03}, our aim is to first provide a coupling $\left(\bm{\hat{S}_i}(\bar{\alpha},\bar{n}),
\bm{\hat{W}_i}(\bar{\alpha},\bar{n})\right)$ with the corresponding law denoted by
$\hat{\mathbb{P}}$ such that
\begin{itemize}
 \item $\bm{\hat{S}_i}(\bar{\alpha},\bar{n}) \overset{D}{=} \bm{{S}_i}(\bar{\alpha},\bar{n})$ and
 $\bm{\hat{W}_i}(\bar{\alpha},\bar{n}) \overset{D}{=} \bm{{W}_i}(\bar{\alpha},\bar{n})$
\item $\hat{\mathbb{P}}\left(\bm{\hat{S}_i}(\bar{\alpha},\bar{n}) = \bm{\hat{W}_i}(\bar{\alpha},\bar{n}) \right) = 1$
\end{itemize}

The second requirement will help us to show that the two distributions
$\mathbb{P}$ and $\mathbb{P'}$ are equal. This follows from the coupling inequality
 {(see \cite{Lindvall02} for more on coupling inequalities)}
\begin{equation}
\label{eq:coupling}
\left\Vert \mathbb{P} - \mathbb{P'} \right\Vert \leq 2\hat{\mathbb{P}}
\left(\bm{\hat{S}_i}(\bar{\alpha},\bar{n}) \neq \bm{\hat{W}_i}(\bar{\alpha},\bar{n}) \right).
\end{equation}
Such a coupling is precisely obtained as follows.

Consider two tagged Class~$i$ customers $X$ and $Y$ that arrive to
a $MPS(\bar{\alpha},\bar{n})$ and a $MROS(\bar{\alpha},\bar{n})$ system respectively.
This means that at the arrival instant of customer X in the
$MPS(\bar{\alpha},\bar{n})$ system, there are $n_i$ Class~$i$ customers
already present in the system. Similarly, at the arrival
instant of customer Y in $MROS(\bar{\alpha},\bar{n}),$ there are $n_i$
customers of Class~$i$ that are waiting for service in the queue.
Recall that $\bar{\beta} = (\beta_1,\ldots,\beta_N)$ where $\beta_i$ in the MPS system
denotes the number of Class~$i$ customers that are receiving service. In the MROS
system, $\beta_i$ denotes those (waiting) Class~$i$ customers from which
the next customer could be chosen. Note that since $\sum_{i=1}^N n_i = n,$ with the arrival of
customer $X,$ the $MPS(\bar{\alpha},\bar{n})$ system has $n+1$ customers.
Similarly, with the arrival of customer $Y,$ the $MROS(\bar{\alpha},\bar{n})$ system
has $n+2$ customers of which one customer is in service and the remaining
$n+1$ customers (including customer Y) are waiting for service.
We will now specify the rule for forming the required coupling.
Since the customers can be distinguished by their class index and also the position in
their respective queues, we couple the $n+1$ customers in $MPS(\bar{\alpha},\bar{n})$ with the $n+1$
waiting customers in the $MROS(\bar{\alpha},\bar{n})$ system based on their class
and queue position. The coupling must be such that the coupled customers belong
to the same class and invariably have the same queue position in their respective queues.
It goes without saying that the tagged customers $X$ and $Y$ are also coupled.
As in \cite{Borst03}, we also couple the subsequent arriving customers and let $D_1, D_2 \ldots$
denote i.i.d random variables with an exponential distribution of rate $\mu.$
These random variables correspond to service times of a customer in service in
$MROS(\bar{\alpha}).$ At the service completion epoch, pick a pair of coupled customers
randomly. This random picking is with a distribution such that the chosen pair is of Class~$i$ with
probability $\frac{p_i}{\sum_{j=1}^N p_j\beta_j}.$
When the randomly chosen pair is of Class~$i$,  a class $i$
customer departs from the MPS system while such a customer is taken for service in the MROS system.
This process is repeated till the tagged pair $(X,Y)$ leaves the system.
Clearly, this joint probability space is so constructed that the random variables
$\bm{\hat{S}_i}(\bar{\alpha},\bar{n}) = \bm{\hat{W}_i}(\bar{\alpha},\bar{n})$ $\hat{\mathbb{P}}$--a.s.
From Eq. \eqref{eq:coupling}, this implies that
\begin{equation}
\label{eq:cond_distr_equal}
\bm{{S}_i}(\bar{\alpha},\bar{n}) \overset{D}{=} \bm{{W}_i}(\bar{\alpha},\bar{n}).
\end{equation}

Now let random vectors $\bm{N}^{MPS}$ (resp. $\bm{N}_1^{MROS}$) denote the vector of multiclass
customers present in the system in steady state (resp. waiting in the system in steady state in MROS).
The subscript $1$ in $\bm{N}_1^{MROS}$ is used to indicate a busy server.
 {Since the arrival process is Poisson, the unconditional} probabilities are
given by the following
\begin{eqnarray}
 P(\bm{{S}_i}(\bar{\alpha}) > t)  &=& \sum_{\bar{n}} P(\bm{N}^{MPS}=\bar{n})
 P(\bm{{S}_i}(\bar{\alpha},\bar{n})>t). \nonumber \\
 \end{eqnarray}
Similarly, we have
\begin{eqnarray}
 P(\bm{{W}_i}(\bar{\alpha}) > t)  = \sum_{\bar{n}} P(\bm{N}_1^{MROS}=\bar{n})
 P(\bm{{W}_i}(\bar{\alpha},\bar{n})>t).
 \end{eqnarray}

Now if $P(\bm{N}_1^{MROS}=\bar{n}) = \rho P(\bm{N}^{MPS}=\bar{n})$
is true, then from Eq. \eqref{eq:cond_distr_equal},
the statement of the theorem follows and this would complete the proof.
In the following lemma, we shall prove that indeed
$P(\bm{N}_1^{MROS}=\bar{n}) = \rho P(\bm{N}^{MPS}=\bar{n}).$
\end{proof}

\bigskip

\begin{lemma}
 $P(\bm{N}_1^{MROS}=\bar{n}) = \rho P(\bm{N}^{MPS}=\bar{n})$
 for $ \bar{n}$ such that $|\bar{n}| \geq 0.$
\end{lemma}
\begin{proof}
 We first simplify the notations as follows.
Let $\pi(\bar{n}):= P(\bm{N}^{MPS}=\bar{n})$ and
$\hat{\pi}(1,\bar{n}):= P(\bm{N}_1^{MROS}=\bar{n}).$
Let $\hat{\pi}(0,\bar{0})$ denote the probability that
the MROS system has no customers and is idle. The statement of
the lemma now requires us to prove that
$\hat{\pi}(1,\bar{n}) = \rho \pi(\bar{n}).$ To prove this result,
consider the balance equation for the MPS system where
$\pi$ shall denote the stationary invariant distribution for
the system. The assumption $\Lambda < \mu$ implies that the
underlying Markov process is ergodic and hence the stationary
distribution $\pi$ is unique. For $ \bar{n}$ such that
$|\bar{n}| \geq 0,$ the  {global}
balance equations for the MPS$(\bar{\alpha})$ system are
\begin{eqnarray*}
&~&(\Lambda + \sum_{i=1}^N\left(\frac{\beta_i(\bar{n})p_i}{\sum_{j=1}^N p_j \beta_j(\bar{n})}
\right)\mu \mathbbm{1}_{\{|\bar{n}| > 0\}})\pi(\bar{n}) \\
&=& \sum_{i=1}^N \lambda_i \mathbbm{1}_
{\{n_i > 0\}} \pi(\bar{n} - e_i) \nonumber \\
&+&
\sum_{i=1}^N \left(\frac{\beta_i(\bar{n} + e_i)p_i}{\sum_{j=1}^N p_j \beta_j(\bar{n} + e_i)}
\right) \mu\pi(\bar{n} + e_i).
\end{eqnarray*}

Now since $$\sum_{i=1}^N\left(\frac{\beta_i(\bar{n})p_i}{\sum_{j=1}^N p_j \beta_j(\bar{n})}
\right) = 1,$$ the balance equations can be written as

\footnotesize
\begin{eqnarray}
\label{eq:mps}
(\Lambda + \mu \mathbbm{1}_{\{|\bar{n}| > 0\}})\pi(\bar{n}) &=& \sum_{i=1}^N \lambda_i \mathbbm{1}_
{\{n_i > 0\}} \pi(\bar{n} - e_i) \\
&+&
\sum_{i=1}^N \left(\frac{\beta_i(\bar{n} + e_i)p_i}{\sum_{j=1}^N p_j \beta_j(\bar{n} + e_i)}
\right) \mu\pi(\bar{n} + e_i). \nonumber
\end{eqnarray}
\normalsize

Similarly, the  {global} balance equations for the MROS$(\bar{\alpha})$ system are
as follows for $ \bar{n}$ such that $|\bar{n}| \geq 0.$

\begin{eqnarray}
\label{eq:mros}
&~&(\Lambda + \mu\mathbbm{1}_{\{|\bar{n}| > 0\}})\hat{\pi}(1,\bar{n}) = \sum_{i=1}^N \lambda_i \mathbbm{1}_
{\{n_i > 0\}} \hat{\pi}(1,\bar{n} - e_i) \nonumber \\
&+& \sum_{i=1}^N \left(\frac{\beta_i(\bar{n} + e_i)p_i}{\sum_{j=1}^N p_j \beta_j(\bar{n} + e_i)}.
\right) \mu\hat{\pi}(1,\bar{n} + e_i)
\end{eqnarray}

Additionally, the idle system should satisfy
\begin{equation}
 \Lambda \hat{\pi}(0,\bar{0}) = \mu \hat{\pi}(1,\bar{0})
\end{equation}
where $\hat{\pi}(0,\bar{0}) = 1- \rho$ is the probability
that the system is empty.
Now again, the assumption $\Lambda < \mu$ implies that the
underlying Markov process is ergodic and hence the stationary
distribution $\hat{\pi}$ is also unique. Therefore to prove the lemma,
it is sufficient to check if the  {global} balance equations for the
MROS system given by Eq. \eqref{eq:mros} are satisfied when
$\hat{\pi}(1,\bar{n}) = \rho \pi(\bar{n}).$

Now from Eq. \eqref{eq:mros} and assuming that
$\hat{\pi}(1,\bar{n}) = \rho \pi(\bar{n}),$ we have
\begin{eqnarray*}
\label{eq:mros2}
(\Lambda + \mu\mathbbm{1}_{\{|\bar{n}| > 0\}})\hat{\pi}(1,\bar{n})  - \sum_{i=1}^N \lambda_i \mathbbm{1}_
{\{n_i > 0\}} \hat{\pi}(1,\bar{n} - e_i) \nonumber \\
- \sum_{i=1}^N \left(\frac{\beta_i(\bar{n} + e_i)p_i}{\sum_{j=1}^N p_j \beta_j(\bar{n} + e_i)}
\right) \mu\hat{\pi}(1,\bar{n} + e_i) \nonumber \\
 = (\Lambda + \mu\mathbbm{1}_{\{|\bar{n}| > 0\}})\rho{\pi}(\bar{n})  - \sum_{i=1}^N \lambda_i \mathbbm{1}_
{\{n_i > 0\}} \rho{\pi}(\bar{n} - e_i) \nonumber \\
- \sum_{i=1}^N \left(\frac{\beta_i(\bar{n} + e_i)p_i}{\sum_{j=1}^N p_j \beta_j(\bar{n} + e_i)}
\right) \mu\rho{\pi}(\bar{n} + e_i) = 0. \nonumber \\
\end{eqnarray*}

The last equality follows from Eq. \eqref{eq:mps} after dividing throughout by $\rho.$
Similarly,
\begin{eqnarray}
\label{eq:mros3}
 \Lambda \hat{\pi}(0,\bar{0})  - \mu \hat{\pi}(1,\bar{0}) & = & \Lambda \hat{\pi}(0,\bar{0})  - \mu \rho{\pi}(\bar{0}) \nonumber \\
  &=& \mu \left(\rho \hat{\pi}(0,\bar{0})  -  \rho{\pi}(\bar{0})\right)\nonumber \\
   &=& \mu \left(\rho \hat{\pi}(0,\bar{0})  -  \rho(1 - \rho)\right) \nonumber \\
   &=&0.
\end{eqnarray}
Here the third equality is from the fact that $\pi(\bar{0}) = (1-\rho)$
is the probability that the MPS$(\bar{\alpha})$ system is empty.
Clearly, substituting $\hat{\pi}(1,\bar{n}) = \rho \pi(\bar{n}),$
satisfies the balance equations for the MROS system. Since $\hat{\pi}$
is the unique invariant distribution, the statement of the lemma follows.
\end{proof}

\bigskip

We now have the following corollary that establishes the desired
equivalence between DPS (GPS) and DROS (resp. WFQ) scheduling policies.
Note that the result is true only for the case when all customers
have identically distributed service requirements. The equivalence
result is not true in general when the customer classes
differ in their service requirements.

\bigskip

\begin{corollary} \ \\
 \begin{itemize}
  \item $\rho P(\bm{{S}_i}(\bar{\infty}) > t) =  P(\bm{{W}_i}(\bar{\infty}) > t)$ \\where $\bm{{S}_i}(\bar{\infty})$
  denotes the sojourn time of a Class~$i$ customer in DPS system and $\bm{{W}_i}(\bar{\infty})$
  denotes the waiting time of a Class~$i$ customer in DROS system.
  \item $\rho P(\bm{{S}_i}(\bar{e}) > t) =  P(\bm{{W}_i}(\bar{e}) > t)$
  \\where $\bm{{S}_i}(\bar{e})$
  denotes the sojourn time of a Class~$i$ customer in GPS system and $\bm{{W}_i}(\bar{e})$
  denotes the waiting time of a Class~$i$ customer in WFQ system.
 \end{itemize}
\end{corollary}

\section{Discussion}
In this paper, we have proposed two multiclass policies, namely MPS
and MROS, that generalize some important multiclass policies from the literature.
Our policies are parameterized by a vector $\bar{\alpha}$ that can be used to control
performance metrics like the mean delay or mean waiting time per class.
Restricting to the special case where the multiclass customers arrive according to a
Poisson process and have indepndent and exponential service requirements,
we show that the tail of the sojourn time distribution for a class $i$ customer in a
system with the MPS policy is a constant multiple of the tail of the waiting time
distribution of a class $i$ customer in a system with the MROS policy.
As special cases, we have thus proved the above equivalence between DPS (GPS) and DROS
(resp. WFQ) scheduling policies.

 {It is worth mentioning that Borst et al \cite{Borst03} have shown the sojourn time
equivalence between ROS and processor sharing for a more general case when the
arrival process is a general renewal process. While our analysis for MPS and MROS
assumes a Poisson arrival process, it would be of interest to investigate if our
equivalence result is true when the arrival process is a general renewal process.
This is part of future work.
}

\section*{Acknowledgements}
Both authors would like to thank Dr. Rudesindo Queija from CWI, Amsterdam for
several discussions on this work. Research for the second author is
partially supported by the French Agence Nationale de la Recherche (ANR) through the project
ANR-15-CE25-0004 (ANR JCJC RACON). The authors would also
like to acknowledge the support of CEFIPRA, an Indo-French centre for promotion
of advanced research, specifically grant FC/DST-Inria-2016-01/448.

\end{document}